\documentclass[onecolumn]{IEEEtran}
\usepackage[utf8]{inputenc}
\usepackage{amsmath, amssymb, amsthm, mathtools}
\usepackage{xcolor}

\usepackage{graphics}
\usepackage{hyperref}

\usepackage{tikz, pgfplots}
\usetikzlibrary{positioning, decorations.markings}
\theoremstyle{plain}
\newtheorem{theorem}{Theorem}
\newtheorem{lemma}[theorem]{Lemma}

\theoremstyle{definition}

\theoremstyle{remark}

\renewcommand{\mathbf}{\boldsymbol}
\renewcommand{\leq}{\leqslant}
\renewcommand{\geq}{\geqslant}
\newcommand{\defeq}{\triangleq}

\newcommand{\norm}[1]{\left\lVert#1\right\rVert}
\newcommand{\size}[1]{\left\lvert#1\right\rvert}

\newcommand{\abs}[1]{\left\lvert#1\right\rvert}
\newcommand{\set}[1]{\mathcal{#1}}

\newcommand{\rv}[1]{\mathsf{#1}}
\DeclareMathOperator{\argmin}{argmin}
\DeclareMathOperator{\argmax}{argmax}
\DeclareMathOperator{\D}{d} % 'd' for differential 
\DeclareMathOperator{\Ker}{ker}
\DeclareMathOperator{\Span}{span}
\DeclareMathOperator{\Cone}{cone}

\newcommand{\trans}{\mathsf{T}}
\DeclarePairedDelimiterX{\infdiv}[2]{(}{)}{#1\delimsize\Vert#2}
\DeclarePairedDelimiterX{\inner}[2]{\langle}{\rangle}{#1,#2}
\makeatletter
\newcommand\ie{\textit{i.e.}}
\newcommand\eg{\textit{e.g.}}
\newcommand\wrt{w.r.t.~}
\newcommand\aka{a.k.a.~}
\makeatother

\newcommand{\hide}[1]{}   % Used for hide large section of codes
\begin{document}

\title{Comments on ``Channel Coding Rate in the Finite Blocklength Regime'': On the Quadratic Decaying Property of the Information Rate Function}

\author{\IEEEauthorblockN{Michael X. Cao, Marco Tomamichel}\\
\IEEEauthorblockA{Department of Electrical and Computer Engineering \& Centre for Quantum Technologies,\\
National University of Singapore}}
\maketitle

%***********************************************************
\begin{abstract}
The quadratic decaying property of the information rate function states that given a fixed conditional distribution $p_{\rv{Y}|\rv{X}}$, the mutual information between the (finite) discrete random variables $\rv{X}$ and $\rv{Y}$ decreases at least quadratically in the Euclidean distance as $p_\rv{X}$ moves away from the capacity-achieving input distributions. It is a property of the information rate function that is particularly useful in the study of higher order asymptotics and finite blocklength information theory, where it was already implicitly used by Strassen~\cite{strassen1962asymptotische} and later,  more explicitly, by Polyanskiy--Poor--Verd\'u~\cite{polyanskiy2010channel}.
However, the proofs outlined in both works contain gaps that are nontrivial to close.
This comment provides an alternative, complete proof of this property.
\end{abstract}

%***********************************************************
\section{Background}
Let $W_{\rv{Y}|\rv{X}}\in\set{P}(\set{Y}|\set{X})$ be a discrete memoryless channel (DMC) from $\set{X}$ to $\set{Y}$ where $\set{X}$ and $\set{Y}$ are some finite sets.
The \emph{information rate function} of $W_{\rv{Y}|\rv{X}}$, \ie,
\begin{align}
 \begin{aligned}[c] I_W: \quad \set{P}(\set{X}) &\to \mathbb{R}_{\geq 0}\\
      p_\rv{X} &\mapsto D\infdiv*{p_\rv{X}\cdot W_{\rv{Y}|\rv{X}}}{p_\rv{X}\times \sum_{\tilde{x}\in\set{X}}W_{\rv{Y}|\rv{X}}(\cdot|x)\cdot p_\rv{X}(\tilde{x})}.\end{aligned}
\end{align}
is a common and useful function in the information theory.
It is a building block of many information quantities of interest; and on its own, $I_W(p_\rv{X})$ equals the asymptotic reliable communication rate across i.i.d. copies of $W_{\rv{Y}|\rv{X}}$ with i.i.d. input source each distributed according to $p_\rv{X}$.
It is well-known that $\sup_{p_\rv{X}\in\set{P}(\set{X})} I_W(p_\rv{X}) = \max_{p_\rv{X}\in\set{P}(\set{X})}I_W(p_\rv{X}) = C(W_{\rv{Y}|\rv{X}})$ is the capacity of $W_{\rv{Y}|\rv{X}}$.
We denote the set of all maximizers (\aka the capacity-achieving input distributions) by $\Pi$.
It is also well-known that $\Pi$ is compact.
As a result, for each $p_\rv{X}\in\set{P}(\set{X})$, the following is well-defined
\begin{equation}\label{eq:def:ppi}
p_\rv{X}^\Pi\defeq\argmin_{p_\rv{X}^\Pi\in\Pi}\norm{p_\rv{X}-p_\rv{X}^\Pi}_2 ,
\end{equation}
where $\norm{\cdot}_2$ denotes the Euclidean norm.
One of the properties of the information rate function is that $I_W(p_\rv{X})$ decreases at least quadratically in the Euclidean distance as $p_\rv{X}$ moves away from the set of capacity-achieving input distributions.
This property, as formalized as the following theorem,  is particularly useful in finite-blocklength information theory (\eg, see~\cite[Theorem 48]{polyanskiy2010channel} and~\cite[Lemma~7]{tomamichel2013tight}.
\begin{theorem}\label{thm:main}
Given a DMC $W_{\rv{Y}|\rv{X}}\in\set{P}(\set{Y}|\set{X})$ where $\set{X}$ and $\set{Y}$ are some finite sets, there exist $\alpha, \mu >0$ such that
\begin{equation}\label{eq:main}
I_W(p_\rv{X}) \leq C(W_{\rv{Y}|\rv{X}}) - \alpha\cdot\norm{p_\rv{X}-p_\rv{X}^\Pi}_2^2    
\end{equation}
for all $p_\rv{X}\in\Pi_\mu\defeq\left\{p_\rv{X}\in\set{P}(\set{X})\middle\vert \norm{p_\rv{X}-p_\rv{X}^\Pi}_2\leq\mu\right\}$.
\end{theorem}
To the best of our knowledge, the statement of this theorem first appeared as one of the intermediate steps in~\cite{strassen1962asymptotische}\footnote{Although quadratic decay is not stated in~\cite{strassen1962asymptotische} explicitly, it follows from (4.41) together with a uniform bound on the third-order term in the Taylor series expansion.}, and was later stated more explicitly and re-proven in~\cite[Page~80]{polyanskiy2010thesis} and~\cite[Page~2351]{polyanskiy2010channel}.
Both proofs considered the Taylor series expansion of the function $I_W$ at points on $\Pi$, \ie, (see~\cite[Eq~(4.39)]{strassen1962asymptotische} and~\cite[Eq~(498)]{polyanskiy2010channel}) 
\begin{align}\label{eq:Taylor}
I_W(p_\rv{X}) &\leq I_W(p_\rv{X}^\Pi) + \nabla I(p^\Pi_\rv{X})^\trans \cdot (p_\rv{X}-p^\Pi_\rv{X}) + \frac{1}{2}(p_\rv{X}-p^\Pi_\rv{X})^\trans \cdot \mathcal{H}(p^\Pi_\rv{X}) \cdot (p_\rv{X}-p^\Pi_\rv{X}) + \beta\cdot \norm{p_\rv{X}-p^\Pi_\rv{X}}_2^3
\end{align}
for some constant $\beta>0$ that only depends on the channel $W_{\rv{Y}|\rv{X}}$\footnote{
Actually, in~\cite[Eq~(4.39)]{strassen1962asymptotische}, the author provided a Taylor series expansion of the sum of $I_W$ and the information variance function.
In~\cite[Eq~(498)]{polyanskiy2010channel}, the tail term is given as $o(\norm{p_\rv{X}-p^\Pi_\rv{X}}_2^2)$ instead of some explicit uniform bound.
}.
Both proofs then argued that (see~\cite[Eq~(4.41)]{strassen1962asymptotische} and~\cite[Eq~(501)]{polyanskiy2010channel}) the term
\begin{equation}\label{eq:fs}
\Phi(p_\rv{X}) \defeq \frac{\nabla I(p^\Pi_\rv{X})^\trans \cdot (p_\rv{X}-p^\Pi_\rv{X}) + \frac{1}{2}(p_\rv{X}-p^\Pi_\rv{X})^\trans \cdot \mathcal{H}(p^\Pi_\rv{X}) \cdot (p_\rv{X}-p^\Pi_\rv{X})}{\norm{p_\rv{X}-p^\Pi_\rv{X}}_2^2}
\end{equation}
is negative and bounded away from $0$ for all $p_\rv{X}\in\Pi_\mu\setminus\Pi$ for some $\mu>0$.
Both proofs correctly pointed out that both terms $\nabla I(p^\Pi_\rv{X})^\trans \cdot (p_\rv{X}-p^\Pi_\rv{X})$ and $(p_\rv{X}-p^\Pi_\rv{X})^\trans \cdot \mathcal{H}(p^\Pi_\rv{X}) \cdot (p_\rv{X}-p^\Pi_\rv{X})$ are non-positive, and cannot be zero simultaneously unless $p_\rv{X}=p_\rv{X}^\Pi$ (see~\cite[Lemma~3.2 and the arguments after (4.43)]{strassen1962asymptotische} and~\cite[below~(498)]{polyanskiy2010channel}): If we assume both terms to be zero, then the second term being zero will force $p_\rv{X}-p^\Pi_\rv{X}\in\Ker(W_{\rv{Y}|\rv{X}})$, in which case the function $I_W$ is linear along the direction from $p^\Pi_\rv{X}$ to $p_\rv{X}$, which, combined with the hypothesis that the first term being zero, will force $I_W(p_\rv{X})=I_W(p^\Pi_\rv{X})$.
This is a key observation, and we also use this exact same argument in our proof.
In other words, it has been fully established in the previous works that~$\Phi$ is strictly negative for all $p_\rv{X}\in\Pi_\mu\setminus\Pi$.
However, we are disputing that the previous proofs are sufficient to show that~$\Phi$ is \emph{bounded away} from zero.
We will exhibit the gaps in the arguments in~\cite{strassen1962asymptotische} and~\cite{polyanskiy2010channel} as follows.

In~\cite{strassen1962asymptotische}, it is argued that, if $\Phi$ is not bounded away from zero, then there must exist a sequence $\{p_k\}_{k=1}^\infty\in\Pi_\mu\setminus\Pi$ such that
\begin{enumerate}
   \item $\limsup_{k\to\infty}\Phi(p_k)\geq 0$;
   \item $\norm{p_k - p_k^\Pi}_2=\mu$ for all $k$ as a result of~(4.38), \ie, $\nabla I(p^\Pi_\rv{X})^\trans \cdot (p_\rv{X}-p^\Pi_\rv{X}) \leq 0$ for all $p_\rv{X}\in\Pi_\eta$.
\end{enumerate}
We agree with the existence of such a sequence satisfying 1), but we fail to see how 2) can also be met.
Indeed, using~(4.38), one can rewrite $\limsup_{k\to\infty}\Phi(p_k)\geq 0$ as
\begin{equation}
\limsup_{k\to\infty}\frac{(p_k-p^\Pi_k)^\trans \cdot \mathcal{H}(p^\Pi_k) \cdot (p_k-p^\Pi_k)}{\norm{p_k-p^\Pi_k}_2^2} = \limsup_{k\to\infty} \left(\frac{p_k-p^\Pi_k}{\norm{p_k-p^\Pi_k}}\right)^\trans\cdot\mathcal{H}(p^\Pi_k)\cdot\left(\frac{p_k-p^\Pi_k}{\norm{p_k-p^\Pi_k}}\right) \geq 0
\end{equation}
which enables one to replace the sequence $\{p_k\}_k$ by a normalized version $\tilde{p}_k\gets p_k^\Pi+\mu\cdot\norm{p_k-p_k^\Pi}_2^{-1}\cdot(p_k-p_k^\Pi)$.
However, the sequence $\{\tilde{p}_k\}_k$ constructed this way may protrude outside the probability polytope $\set{P}(\set{X})$ (see~Fig.~\ref{fig:strassen:counter:example}).
\begin{figure}\centering
\begin{minipage}[t]{.49\textwidth}
\centering
\begin{tikzpicture}[x=5cm, y=4.330cm, dot/.style={draw, fill=black, circle, inner sep=0pt, outer sep=0pt, minimum size=1pt}, decoration={markings, mark= at position 0.5 with {\arrow{stealth}}}]
    \draw[fill=gray!20,draw=none] (.091,.182) -- (.291,.582) -- (.709,.582) -- (.909,.182) -- (.1,.182);
    \node[anchor=south east, gray] at (.9,.182) {$\Pi_\mu$};
    \draw (0,0) -- (1,0) -- (.5,1) -- (0,0);
    \node[anchor=south west] at (1,0) {$\set{P}(\set{X})$};
    \draw[red] (.191,.382) -- (.809,.382); 
    \node[anchor=east, red] at (.191,.382) {$\Pi$};
    \draw[blue, out=180, in=40, postaction={decorate}](.7, .55) to (.191,.382);
    \draw[blue, dashed, postaction={decorate}](.7, .582) -- (.191, .582);
\end{tikzpicture}
\caption{Exposition of the problem with the proof in~\cite{strassen1962asymptotische}: The statement of~\eqref{eq:fs} \emph{not} being bounded from $0$ is \emph{equivalent} to the existence of some sequence $\{p_k\}_{k=1}^\infty\in\Pi_\mu\setminus\Pi$ such that $\limsup_{k\to\infty}\Phi(p_k)\geq 0$ (see the blue line). However, the ``shape'' of the probability polytope $\set{P}(\set{X})$ may prevent such a sequence to be normalized (see the dashed line).}
\label{fig:strassen:counter:example}
\end{minipage}
~
\begin{minipage}[t]{.49\textwidth}
\centering
%\begin{tikzpicture}[x=5cm, y=4.330cm, dot/.style={draw, fill=black, circle, inner sep=0pt, outer sep=0pt, minimum size=1pt}, decoration={markings, mark= at position 0.5 with {\arrow{stealth}}}]
%    \draw[fill=gray!20,draw=none] (.091,.182) -- (.291,.582) -- (.709,.582) -- (.909,.182) -- (.1,.182);
%    \node[anchor=south east, gray] at (.9,.182) {$\Pi_\mu$};
%    \draw (0,0) -- (1,0) -- (.5,1) -- (0,0);
%    \node[anchor=south west] at (1,0) {$\set{P}(\set{X})$};
%    \draw[red] (.191,.382) -- (.809,.382); 
%    \node[anchor=east, red] at (.191,.382) {$\Pi$};
%    \draw[blue] (.67,.582) -- (.25,.382);
%    \draw[blue, dashed] (.67,.582) -- (.67,.382);
%    \draw[blue] (.35,.382) arc (0:28:.1); \node at (.4,.415) {$\theta$};
%    \node[dot, blue] at (.67,.582) {}; \node[dot, blue] at (.25,.382) {}; \node[dot, blue] at (.67,.382) {};
%\end{tikzpicture}
%\caption{An adaptation of the arguments in~\cite{strassen1962asymptotische} by considering \emph{all} line segments any points that are $\mu$-away from $\Pi$ to any points on $\Pi$ with additional geometric and compactness arguments.}
%\label{fig:strassen:potential:fix}
\includegraphics[height=4.4cm]{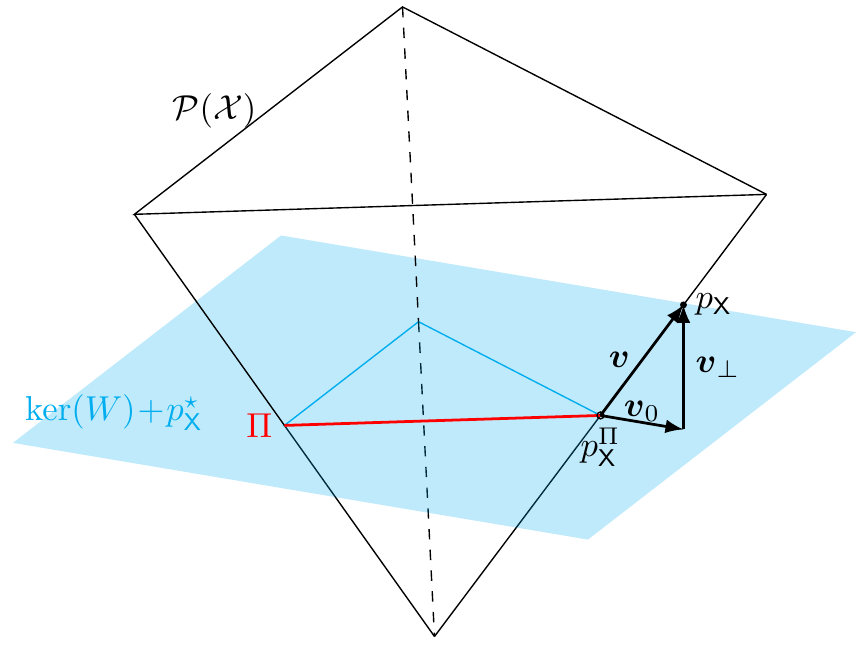}
\caption{Exposition of the problem with the proof in~\cite{polyanskiy2010channel}: This depicts a situation when $\mathbf{v}_0$ is pointing outside of the probability polytope (see~\eqref{eq:disagree}).}
\label{fig:polyanskiy:counter:example}
\end{minipage}
\end{figure}

In~\cite{polyanskiy2010channel}, the authors decompose $p_\rv{X}-p^\Pi_\rv{X}=v_0+v_\perp$ where $v_0\in\Ker(W_{\rv{Y}|\rv{X}})$ and $v_\perp\in\Ker(W_{\rv{Y}|\rv{X}})^\perp$, and claim (see~\cite[(498)--(505)]{polyanskiy2010channel})
\begin{align}
\Phi(p_\rv{X})\cdot \norm{p_\rv{X}-p^\Pi_\rv{X}}_2^2 %= \Phi(p_\rv{X})\cdot \norm{v_0+v_\perp}_2^2
&= \nabla I(p^\Pi_\rv{X})^\trans \cdot (v_0+v_\perp) + \frac{1}{2}v_\perp^\trans \cdot \mathcal{H}(p^\Pi_\rv{X}) \cdot v_\perp \\
\label{eq:unclear}
&\leq \nabla I(p^\Pi_\rv{X})^\trans \cdot v_0 + \frac{1}{2}v_\perp^\trans \cdot \mathcal{H}(p^\Pi_\rv{X}) \cdot v_\perp \\
\label{eq:disagree}
&\leq  -\Lambda_1 \cdot \norm{v_0}_2 - \frac{\lambda_{\min+}(W_{\rv{Y}|\rv{X}}W_{\rv{Y}|\rv{X}}^\trans)}{2\cdot q_{\min}} \cdot \norm{v_\perp}_2^2\\
&\leq  -\min\left\{\Lambda_1, \frac{\lambda_{\min+}(W_{\rv{Y}|\rv{X}}W_{\rv{Y}|\rv{X}}^\trans)}{2\cdot q_{\min}}\right\} \cdot \norm{p_\rv{X}-p^\Pi_\rv{X}}_2^2
\end{align}
where $\Lambda_1$ is some positive number (independent from $p_\rv{X}$) and $\lambda_{\min+}(\cdot)$ stands for the smallest strictly positive eigenvalue of a matrix.
We particularly disagree with~\eqref{eq:disagree}, as it may occur that $p^\Pi_\rv{X}+t\cdot v_0\not\in\set{P}(\set{X})$ for all $t>0$ (see Fig.~\ref{fig:polyanskiy:counter:example}).
In such cases, there is no guarantee that $\nabla I(p^\Pi_\rv{X})^\trans \cdot v_0\leq 0$ as $I_W$ may actually increase along the direction of $v_0$ from $p_\rv{X}^\Pi$.
% I am thinking deleting the following two sentences since they are not technically necessary.
%Even if one assumes $p^\Pi_\rv{X}+t\cdot v_0\in\set{P}(\set{X})$ for all $t>0$ small enough, it is unclear to us why the term $\nabla I(p^\Pi_\rv{X})^\trans \cdot v_0$ is bounded away from $0$.
%Additionally, it is unclear to us why $\nabla I(p^\Pi_\rv{X})^\trans \cdot v_\perp\leq 0$ which is used in~\eqref{eq:unclear}.

One potential way to close the gap in the proof above involves careful observations of the shape of the probability polytope, especially its implication of the relationship between the vector $(p_\rv{X}-p_\rv{X}^\Pi)$ and the affine subspace that $\Pi$ lies in when $p_\rv{X}^\Pi$ is at the boundary of $\Pi$.
Namely, we argue that the angle between the aforementioned vector and the affine subspace is bounded away from $0$.
However, arguments along this line of thinking turned out to be tedious.
In the next section, we present an alternative concise proof that was partially inspired by this very observation. 
%instructive 

%***********************************************************
\section{An alternative proof} \label{sec:alt:proof}

In this section, we present an alternative proof that bypasses the problem we pointed out in the last section.
Same as the previous proofs, we start with~\eqref{eq:Taylor}, \ie, 
\begin{equation}
I_W(p^\Pi_\rv{X}) - I_W(p_\rv{X}) = -\Phi(p_\rv{X})\cdot \norm{p_\rv{X}-p^\Pi_\rv{X}}_2^2 + o(\norm{p_\rv{X}-p^\Pi_\rv{X}}_2^2).
\end{equation}
As already established in the previous proofs, we know $\Phi(p_\rv{X})<0$ for any pmf $p_\rv{X}\not\in\Pi$. 
The technical difficulty is to bound $\Phi(p_\rv{X})$ away from $0$ and the term ${o(\norm{p_\rv{X}-p^\Pi_\rv{X}}_2^2)}/{\norm{p_\rv{X}-p^\Pi_\rv{X}}_2^2}$ close to $0$ for all $p_\rv{X}$ in some vicinity of $\Pi$.
To achieve that, our proof follows the following steps:
\begin{enumerate}
    \item Construct some function $\alpha(p_\rv{X})>0$ such that for $-\Phi(p_\rv{X})\geq 2\alpha(p_\rv{X})$ where $\alpha$ only depends on the direction from $p_\rv{X}^\Pi$ to $p_\rv{X}$, \ie, $(p_\rv{X}-p_\rv{X}^\Pi)/\norm{p_\rv{X}-p_\rv{X}^\Pi}_2$.
    \item Bound term ${o(\norm{p_\rv{X}-p^\Pi_\rv{X}}_2^2)}/{\norm{p_\rv{X}-p^\Pi_\rv{X}}_2^2}$ such that it only depends on $\norm{p_\rv{X}-p_\rv{X}^\Pi}_2$.
    Combined with 1), this means that
    \begin{equation}\label{eq:idea}
    I_W(p^\Pi_\rv{X}) - I_W(p^\Pi_\rv{X}+t\cdot d_\rv{X}) \geq 2\alpha(d_\rv{X})\cdot t^2 + f(t)\cdot t^2
    \end{equation}
    where $d_\rv{X}$ is a norm-1 vector such that $(p^\Pi_\rv{X}+d_\rv{X})^\Pi=p^\Pi_\rv{X}$, and where $f$ is some continuous function with $f(0)=0$.
    \item One can show that the set of all possible $d_\rv{X}$ in~\eqref{eq:idea}, \ie, all valid directions from a point on $\Pi$ to another point in $\set{P}(\set{X})$ with the former point being the closest point on $\Pi$ \wrt the latter, is closed.
    The idea is that the cone of all valid directions from any fixed point $p_\rv{X}^\Pi$ on $\Pi$ is closed, and there are finitely many different such cones as $p_\rv{X}^\Pi$ moves across $\Pi$.
    (See Lemma~\ref{lem:cone:closed}.)
    \item Using the closedness of the set of $d_\rv{X}$, one can find a smallest $\alpha^\star>0$ such that~\eqref{eq:idea} is valid for all $d_\rv{X}$.
    Then, using the continuity of $f$ and the fact that $f(0)=0$, one can find $\mu$ such that $\abs{f(t)}\leq\alpha^\star$ for all $t\in[0,
    \mu]$.
    Thus,
    \begin{equation}
    I_W(p^\Pi_\rv{X}) - I_W(p^\Pi_\rv{X}+t\cdot d_\rv{X}) \geq \alpha^\star\cdot t^2
    \end{equation}
    for all valid direction $d_\rv{X}$ and all $t\in[0,\mu]$.
\end{enumerate}

Before proceeding to the proof of the theorem, we first treat Step 3) above as follows.
\begin{lemma}\label{lem:cone:closed}
    Let $\set{X}$ be a finite set.
    Let $\Pi$ be the set of capacity achieving distributions of some DMC with input alphabet $\set{X}$.
    Recall that $p_\rv{X}^\Pi$ is defined in~\eqref{eq:def:ppi}.
    Then the cone $\Cone\left(\left\{p_\rv{X}-p_\rv{X}^\Pi \,\middle|\, p_\rv{X} \in\set{P}(\set{X}) \right\}\right)$ is closed.
\end{lemma}
\begin{proof}
    For any point $p_\rv{X}^\star\in\Pi$, define
    \begin{equation}\label{eq:def:J:p}
    \set{J}(p_\rv{X}^\star) \defeq \left\{d_\rv{X}\in\mathbb{R}^\set{X} \,\middle\vert\, 
    p_\rv{X}^\star + d_\rv{X} \in\set{P}(\set{X})\ \land \ \left(p_\rv{X}^\star + d_\rv{X}\right)^\Pi = p_\rv{X}^\star \right\}
    \end{equation}
        so that $\left\{p_\rv{X}-p_\rv{X}^\Pi \,\middle|\, p_\rv{X} \in\set{P}(\set{X}) \right\} = \bigcup_{p_\rv{X}^\star\in\Pi} \set{J}(p_\rv{X}^\star)$.
   We also introduce the subspace $\set{V}\defeq\Span(\Pi-p^\star_\rv{X})$ of vectors with components summing to zero and parallel to $\Pi$, which is independent of the choice of $p_\rv{X}^\star\in\Pi$. This allows us to write $\Pi=(p^\star_\rv{X}+\set{V})\cap\mathbb{R}^\set{X}_{\geq 0}$.
   %\footnote{This can be shown as follows: For any two points $p_0^\star$ and $p_1^\star\in\Pi$, we first have $\Span(\Pi-p_1^\star)=\Span((\Pi-p_0^\star)+(p_0^\star-p_1^\star)) \subset\Span(\Pi-p_0^\star)$ since $p_0^\star-p_1^\star$ is a point in $\Span(\Pi-p_0^\star)$. By symmetry, we also have $\Span(\Pi-p_0^\star)\subset\Span(\Pi-p_1^\star)$. Thus, it must hold that $\Span(\Pi-p_0^\star)=\Span(\Pi-p_1^\star)$.}$
   
    We first want to argue that $\Cone( \set{J}(p_\rv{X}^\star))$ is only a function of the support of $p_\rv{X}^\star$, or equivalently the set $\set{X}_0(p_\rv{X}^\star) \defeq \{x\in\set{X}\, \vert\, p_\rv{X}^\star(x) = 0\}$. To see this, we first write
    \begin{align}\label{eq:J:p}
    \set{J}(p_\rv{X}^\star) &= \left\{d_\rv{X}\in\mathbb{R}^\set{X}\, \middle\vert\ \begin{aligned}
    & \sum_{x\in\set{X}} d_\rv{X}(x) = 0 \\
    & d_\rv{X}(x)+p_\rv{X}^\star(x) \geq 0 \quad \forall x\in\set{X}\\
    & \inner{d_\rv{X}}{v} \leq 0 \quad \forall v\in\set{V}\text{ s.t. }v(x)\geq 0 \ \forall x\in\set{X}_0(p_\rv{X}^\star)
    \end{aligned}\right\} .
    \end{align}
    To arrive at the last condition we first use the fact that, for any $p_\rv{X}^\star + d_\rv{X}\in\set{P}(\set{X})$, $\left(p_\rv{X}^\star + d_\rv{X}\right)^\Pi = p_\rv{X}^\star$ if and only if $\inner{d_\rv{X}}{v}\leq 0$ for any vector $v\in\mathbb{R}^\set{X}$ such that $p^\star_\rv{X}+v\in\Pi$, or, in other words, $v\in\Pi-p^\star_\rv{X}$. Since the condition is invariant under multiplication by a scalar, we can extend this to the cone $v \in \set{V}$ with $v(x)\geq 0$ for all $x\in\set{X}_0(p_\rv{X}^\star)$.
    
    Next, let us look at the cone of $\set{J}(p_\rv{X}^\star)$. We can see that the second condition, $d_\rv{X}(x)+p_\rv{X}^\star(x) \geq 0$ for all $x\in\set{X}$ simplifies to $d_\rv{X}(x)\geq 0$ for all $x\in\set{X}_0(p_\rv{X}^\star)$. Moreover, the third condition simply requires $-d_\rv{X}$ to be in the dual cone of $\{ v \in \set{V} \,|\, v(x) \geq 0\ \forall x \in \set{X}_0(p_\rv{X}^\star) \}$. We can thus write
    \begin{align}\label{eq:cone:J:p}
    \Cone\left(\set{J}(p_\rv{X}^\star)\right) &= \left\{d_\rv{X}\in\mathbb{R}^\set{X} \middle\vert\begin{aligned}
    & \sum_{x\in\set{X}} d_\rv{X}(x) = 0 \\
    & d_\rv{X}(x)\geq 0 \quad \forall x\in\set{X}_0(p_\rv{X}^\star)
    \end{aligned}\right\} \cap -\left\{v\in\set{V}\middle\vert v(x)\geq 0 \ \forall x\in\set{X}_0(p_\rv{X}^\star)\right\}^*
    \end{align}
    This cone is closed since it is the intersection of a closed cone with the dual of a cone, which is also closed.
    It evidently also only depends on $p_\rv{X}^\star$ via its support, \ie, via $\set{X}_0(p_\rv{X}^\star)$.
    It then remains to notice that
    \begin{align}\label{eq:cone:J}
        \Cone\left(\left\{p_\rv{X}-p_\rv{X}^\Pi \,\middle|\, p_\rv{X} \in\set{P}(\set{X}) \right\}\right) = \bigcup_{p_\rv{X}^\star\in\Pi} \Cone\left(\set{J}(p_\rv{X}^\star)\right)
    \end{align}
    is a finite union of closed cones, and is thus also closed.
\end{proof}

\begin{proof}[Proof of Theorem~\ref{thm:main}]
Without loss of generality, we assume all output symbols of the channel $W_{\rv{Y}|\rv{X}}$ is reachable by some input, \ie, there does not exists $y\in\set{Y}$ such that $W_{\rv{Y}|\rv{X}}(y|x)=0$ for all $x\in\set{X}$.
Otherwise, one could simply exclude such $y$ from the output alphabet.
This ensures that $q^\star_\rv{Y}$, the unique (see, \eg,~\cite{shannon1957some}) capacity-achieving output distribution of $W_{\rv{Y}|\rv{X}}$, has full support.

Firstly, we consider the Taylor series expansion of the function $I_W$ in a neighborhood of $p_\rv{X}^\star\in\Pi$ along a direction $d_\rv{X}\in\mathbb{R}^{\set{X}}$ with $\norm{d_\rv{X}}_2=1$ such that $p_\rv{X}^\star+t\cdot d_{\rv{X}} \in\set{P}(\set{X})$ for all $t>0$ sufficiently small, \ie,
\begin{equation}\label{eq:Taylor:series:expansion}
I_W(p_\rv{X}^\star+t\cdot d_\rv{X}) = C(W_{\rv{Y}|\rv{X}}) + (\nabla I^\trans \cdot d_\rv{X})\cdot t + (d_\rv{X}^\trans\cdot H \cdot d_\rv{X}) \cdot t^2 + \left(\sum_{n=1}^\infty \frac{(-1)^{n+1}}{(n+2)(n+1)} \cdot \sum_{y\in\set{Y}} \frac{d_\rv{Y}^2(y)}{q^\star_\rv{Y}(y)}\cdot \left[\frac{d_\rv{Y}(y)}{q^\star_\rv{Y}(y)}\right]^n\cdot t^n\right) \cdot t^2
\end{equation}
where, for $y \in \set{Y}$, $x, x_1, x_2 \in \set{X}$,
\begin{align}
d_\rv{Y}(y)&\defeq\sum_{x\in\set{X}} W_{\rv{Y}|\rv{X}}(y|x)\cdot d_\rv{X}(x),\\
\left[\nabla I\right]_x &\defeq \left. \frac{\partial I_W}{\partial p(x)} \right\vert_{p_\rv{X}^\star} =
%\sum_y W_{\rv{Y}|\rv{X}}(y|x) \cdot \log{\frac{W_{\rv{Y}|\rv{X}}(y|x)}{q^\star_\rv{Y}(y)}} - 1
D\infdiv*{W_{\rv{Y}|\rv{X}}(\cdot|x)}{q_\rv{Y}^\star} - 1,
%\qquad \forall x\in\set{X},
\\
\left[H\right]_{x_1, x_2} &\defeq \left.\frac{\partial^2 I_W}{2\cdot\partial p(x_2) \partial p(x_1)}\right\vert_{p_\rv{X}^\star} = -\sum_{y\in\set{Y}} \frac{W_{\rv{Y}|\rv{X}}(y|x_1)\cdot W_{\rv{Y}|\rv{X}}(y|x_2)}{2\cdot q^\star_\rv{Y}(y)}. %\qquad \forall x_1,x_2\in\set{X},\\
\end{align}
Notice that $H$ is a negative semi-definite matrix.
By writing $d_\rv{X}^\trans\cdot H \cdot d_\rv{X} = -\sum_{y}\frac{1}{2\cdot q_\rv{Y}^\star(y)} \cdot \norm{\sum_{x}W_{\rv{Y}|\rv{X}}(\cdot|x)\cdot d_\rv{X}(x)}^2_2$, we see that $d_\rv{X}^\trans\cdot H \cdot d_\rv{X}=0$ if and only if $d_\rv{X}\in\Ker(W_{\rv{Y}|\rv{X}})$.
Observe that the expansion in~\eqref{eq:Taylor:series:expansion} is independent of $p^\star_\rv{X}\in\Pi$.
Moreover, the series in the brackets converges absolutely for all $t\in[0,q_{\min}/\sqrt{\size{\set{X}}})$, where $q_{\min}\defeq\min_{y\in\set{Y}}q^\star_\rv{Y}(y)>0$, since
\begin{equation}\label{eq:bound:third}
\sum_{n=1}^\infty \abs{\frac{(-1)^{n+1}}{(n+2)(n+1)} \cdot \sum_{y\in\set{Y}} \frac{d_\rv{Y}^2(y)}{q^\star_\rv{Y}(y)}\cdot \left[\frac{d_\rv{Y}(y)}{q^\star_\rv{Y}(y)}\right]^n\cdot t^n} \leq \frac{\size{\set{X}}\cdot\size{\set{Y}}}{q_{\min}}\cdot \sum_{n=1}^\infty %\frac{1}{(n+2)(n+1)} \cdot 
\left(\frac{\sqrt{\size{\set{X}}}}{q_{\min}}\cdot t \right)^n = \underbrace{ \frac{\size{\set{X}}\cdot\size{\set{Y}}}{q_{\min}}\cdot \frac{\sqrt{\size{\set{X}}}\cdot t}{q_{\min} - \sqrt{\size{\set{X}}}\cdot t} }_{ =:\ f(t)},
\end{equation}
where we utilize the fact that $\norm{W}_2\defeq\sup_{\norm{v}_2=1}\norm{\sum_{x}W(\cdot|x)\cdot v(x)}_2 \leq \sqrt{\size{\set{X}}}$.

Secondly, as shown in the proofs in~\cite{strassen1962asymptotische} and~\cite{polyanskiy2010thesis,polyanskiy2010channel}, at least one of the quantities $\nabla I^\trans \cdot d_\rv{X}$ and $d_\rv{X}^\trans\cdot H \cdot d_\rv{X}$ is strictly negative.
We repeat the argument here.
Notice that $\nabla I^\trans \cdot d_\rv{X}\leq 0$ (since $p_\rv{X}^\star$ is a maximizer) and $d_\rv{X}^\trans\cdot H \cdot d_\rv{X}\leq 0$ (since $H$ is negative semi-definite).
Assume now further that $d_{\rv{X}}$ is pointing away from $\Pi$, \ie, $p_\rv{X}^\star+t\cdot d_\rv{X}\in\set{P}(\set{X})\setminus\Pi$ for $t>0$ being sufficiently small. We want to argue that either $\nabla I^\trans \cdot d_\rv{X}$ or $d_\rv{X}^\trans\cdot H \cdot d_\rv{X}$ is nonzero.
If $d_\rv{X}^\trans\cdot H \cdot d_\rv{X} =0$ then $d_\rv{X}\in\Ker(W_{\rv{Y}|\rv{X}})$ (as discussed before), and thus $I_W$ changes linearly along the direction of $d_\rv{X}$ (since $d_\rv{Y} = 0$ in this case).
Hence, $I_W(p^\star_\rv{X} + t\cdot d_\rv{X}) = C(W_{\rv{Y}|\rv{X}}) + \nabla I^\trans \cdot d_\rv{X}\cdot t $ for all $t>0$, contradicting with the assumption that $p^\star_\rv{X}+t\cdot d_\rv{X}\not\in\Pi$ unless $I^\trans \cdot d_\rv{X} < 0$.
Therefore, the following function
\begin{equation}
\alpha(d_\rv{X}) \defeq - \frac{1}{2}\cdot\left(\nabla I^\trans\cdot d_\rv{X} + d_\rv{X}^\trans\cdot H \cdot d_\rv{X}\right)
\end{equation}
is strictly positive for all $d_\rv{X}$ pointing away from $\Pi$.

Thirdly, we focus on the set of all \emph{valid} directions, namely 
\begin{align}
    \set{J} &\defeq \left\{\frac{p_\rv{X}-p_\rv{X}^\Pi}{\norm{p_\rv{X}-p_\rv{X}^\Pi}_2} \,\middle\vert\, p_\rv{X}\in\set{P}(\set{X})\setminus\Pi\right\} = \Cone\left( \left\{ p_\rv{X}-p_\rv{X}^\Pi \,\middle\vert\, p_\rv{X}\in\set{P}(\set{X}) \right\} \right) \cap \left\{ d_\rv{X}\in\mathbb{R}^{\set{X}} \, \middle\vert\, \norm{d_\rv{X}}_2=1\right\} \,.
\end{align}
We argue that this set is compact since $\Cone\left( \left\{ p_\rv{X}-p_\rv{X}^\Pi \,\middle\vert\, p_\rv{X}\in\set{P}(\set{X}) \right\} \right)$ is closed (see Lemma~\ref{lem:cone:closed}), and the set $\left\{ d_\rv{X}\in\mathbb{R}^{\set{X}} \, \middle\vert\, \norm{d_\rv{X}}_2=1\right\}$ is bounded and closed.
Since $\alpha$ is a continuous function of $d_\rv{X}$, we may pick
\begin{equation}
\alpha \defeq \min_{d_\rv{X}\in\set{J}} \alpha(d_\rv{X})
\end{equation}
which is strictly positive.
Since $t > t^2$ for $t < 1$, this enables us to write 
\begin{align}
I_W(p_\rv{X}^\star+t\cdot d_\rv{X}) &\leq C(W_{\rv{Y}|\rv{X}}) + (\nabla I^\trans \cdot d_\rv{X})\cdot t^2 + (d_\rv{X}^\trans\cdot H \cdot d_\rv{X}) \cdot t^2 + f(t) \cdot t^2 \leq C(W_{\rv{Y}|\rv{X}}) - \alpha \cdot t^2 + \left( f(t) - \alpha \right) \cdot t^2
\end{align}
for $0 \leq t<q_{\min}/\sqrt{\size{\set{X}}}$.
Notice that $f(0) = 0$ (see~\eqref{eq:bound:third}), and the function is continuous on $[0,q_{\min}/\sqrt{\size{\set{X}}})$.
Thus, there must exist some $\mu\in(0,q_{\min}/\sqrt{\size{\set{X}}})$ such that $f(t) < \alpha$ for all $t\in[0,\mu)$.
Hence, 
\begin{equation}
I_W(p_\rv{X}^\star+t\cdot d_\rv{X}) \leq C(W_{\rv{Y}|\rv{X}}) - \alpha \cdot t^2
\end{equation}
for all $p_\rv{X}^\star\in\Pi$, $t\in[0,\mu)$ and $d_{\rv{X}} \in \mathbb{R}^{\set{X}}$ with $\norm{d_{\rv{X}}}_2 = 1$ and $p_\rv{X}^\star+ t \cdot d_\rv{X} \in \set{P}(\set{X}) \setminus \Pi$.

Finally, to show~\eqref{eq:main}, we substitute $p_\rv{X}^\star\gets p_\rv{X}^\Pi$, $d_\rv{X}\gets \frac{p_\rv{X}-p_\rv{X}^\Pi}{\norm{p_\rv{X}-p_\rv{X}^\Pi}_2}$, and $t\gets \norm{p_\rv{X}-p_\rv{X}^\Pi}_2$ for each $p_\rv{X}\in\Pi_\mu\setminus\Pi$ and note that the inequality is trivial when $p_\rv{X}\in\Pi$.
\end{proof}

\section{Channel with linear input constraints}
In this section, we consider channels with linear constraints on its input probability space. 
In the following, we look at input distribution that can only be chosen from some subset of the probability polytope as
\begin{equation}
\set{P}_\set{A}(\set{X}) \defeq \left\{p\in\set{P}(\set{X}) \ \middle\vert \ \inner{p}{a} \geq 0 \quad\forall a\in\set{A}\right\}
\end{equation}
where $\set{X}$ is the input alphabet, and $\set{A}$ is some finite set of vectors in $\mathbb{R}^\set{\set{X}}$.
Such constraints have practical uses;
for example, an average power constraint can be formulated as 
\begin{equation}
\sum_{x\in\set{X}} p(x) \cdot x^2 \leq P
\end{equation}
where $x^2$ is (proportional to) the amount of power needed for sending the symbol $x$, and $P$ is the available average power per channel use. (This corresponds to the choice $a(x) = P - x^2$.)
The highest rate of realizable transmission over channel $W_{\rv{Y}|\rv{X}}\in\set{P}(\set{Y}|\rv{X})$ under such constraints, \aka the \emph{constrained channel capacity} can be expressed as
\begin{equation}
C_\set{A}(W_{\rv{Y}|\rv{X}}) = \sup_{p_\rv{X}\in\set{P}_\set{A}(\set{X})} I(\rv{X}:\rv{Y})_{p_\rv{X}\cdot W_{\rv{Y}|\rv{X}}}.
\end{equation}
Similar to the quadratic decaying property of $I_W$ in the unconstrained case, the next theorem holds.

\begin{theorem}
Let $W_{\rv{Y}|\rv{X}}\in\set{P}(\set{Y}|\set{X})$ be a DMC from $\set{X}$ to $\set{Y}$ where $\set{X}$ and $\set{Y}$ are some finite sets.
Let $\set{A}$ be a finite subset of $\mathbb{R}^\set{X}$.
There exists some $\alpha, \mu>0$ such that 
\begin{equation}
I_W(p_\rv{X}) \leq C_\set{A}(W_{\rv{Y}|\rv{X}}) - \alpha\cdot\norm{p_\rv{X}-p_\rv{X}^{\Pi^{\set{A}}}}_2^2
\end{equation}
for all $p_\rv{X}\in\Pi^{\set{A}}_\mu$ where the set of {constrained capacity achieving distributions} is $\Pi^\set{A} \defeq \argmax_{p_\rv{X}\in\set{P}_\set{A}(\set{X})} I(\rv{X}:\rv{Y})_{p_\rv{X}\cdot W_{\rv{Y}|\rv{X}}}$
and
%\defeq \left\{p_\rv{X}\in\set{P}_\set{A}(\set{X})\, \middle\vert\, I_W(p_\rv{X})=C_\set{A}(W_{\rv{Y}|\rv{X}}) \right\}$ 
we denote its $\mu$-neighborhood by $\Pi^{\set{A}}_\mu \defeq \big\{p_\rv{X}\in\set{P}_\set{A}(\set{X})\, \big\vert\, \big\| p_\rv{X}-p_\rv{X}^{\Pi^{\set{A}}} \big\|_2 \leq \mu\big\}$.
\end{theorem}

We note that in a concurrent work~\cite{cheng2023mutual} this problem is investigated using different techniques, yielding explicit bounds on the constants $\alpha$ and $\mu$.

The idea of the proof of this theorem is exactly same as that of Theorem~\ref{thm:main} (see points 1)--4) at the beginning of Section~\ref{sec:alt:proof}) with the exception of point 3) where instead of considering all valid directions from points on $\Pi$ to points in $\set{P}(\set{X})$, one considers all valid directions to points in $\set{P}_\set{A}(\set{X})$.
The intuition is that since $\set{P}_\set{A}(\set{X})$ is the resultant polytope after finite number of ``straight cuts'' of $\set{P}(\set{X})$, the closedness argument surrounding point 3) shall still hold.
The proof proceeds analogously to the proof of Theorem~\ref{thm:main}, and thus we simply describe the necessary changes:
\begin{enumerate}
%\setcounter{enumi}{-1}
%\item Some trivial changes: Replace $\set{P}(\set{X})$ by $\set{P}_\set{A}(\set{X})$, and $\Pi$ by $\Pi^\set{A}$.
%Let $q^\star_\rv{Y}$ denote the \emph{unique} (to be justified below) capacity-achieving output distribution in this constrained scenario.
\item We use the notation $p^\star_\rv{X}\in\Pi^\set{A}$ to denote the nearest point in $\Pi^\set{A}$ with regards to $p_\rv{X}$.
\item Thanks to the linearity of the constraints, one can show $\Pi^\set{A}$ to be convex.
In particular, when $\Pi^\set{A}$ is not a singleton, let $p^\star_\rv{X}$ and $\tilde{p}^\star_\rv{X}$ be two (different) points on $\Pi^\set{A}$.
The function $t\mapsto I_W((1-t)\cdot p^\star_\rv{X}+ t\cdot \tilde{p}^\star_\rv{X})$ is concave.
Moreover, since the constraints are linear, all points $(1-t)\cdot p^\star_\rv{X}+ t\cdot \tilde{p}^\star_\rv{X}$ are feasible for $t\in[0,1]$.
This way, the fact that $p^\star_\rv{X}$ and $\tilde{p}^\star_\rv{X}$ are two maximal points forces $I_W((1-t)\cdot p^\star_\rv{X}+ t\cdot \tilde{p}^\star_\rv{X})$ to be constant for all $t\in[0,1]$.
In addition, by writing 
\begin{equation} \begin{aligned}
I_W((1-t)\cdot p^\star_\rv{X}+ t\cdot \tilde{p}^\star_\rv{X}) &= \inner*{(1-t)\cdot p^\star_\rv{X}+ t\cdot \tilde{p}^\star_\rv{X}}{\sum_{y}W_{\rv{Y}|\rv{X}}(y|\cdot)\log{W_{\rv{Y}|\rv{X}}(y|\cdot)}} \\
&\qquad + H\left((1-t)\cdot \sum_{x}W_{\rv{Y}|\rv{X}}(\cdot|x)\cdot p^\star_\rv{X} + t\cdot \sum_{x}W_{\rv{Y}|\rv{X}}(\cdot|x)\cdot \tilde{p}^\star_\rv{X} \right) ,
\end{aligned} \end{equation}
where the first term is linear, and the second term is strictly concave, we can see that the output distribution must stay constant as $t$ changes, \ie, $p^\star_\rv{X}-\tilde{p}^\star_\rv{X}\in\Ker(W_{\rv{Y}|\rv{X}})$. 
This further implies $(1-t)\cdot p^\star_\rv{X}+ t\cdot \tilde{p}^\star_\rv{X}\in\Pi^\set{A}$ as long as it is in the probability polytope and satisfies the linear constraints.
In other words, there exists some subspace $\set{V}\leq\Ker(W_{\rv{Y}|\rv{X}})$ such that $\Pi^\set{A}=\set{V}\cap\set{P}_\set{A}(\set{X})$.
\item We still define $\set{J}(p_\rv{X}^\star)$ in the same way as in~\eqref{eq:def:J:p} (with trivial changes).
However,~\eqref{eq:J:p} need to be rewritten as
    \begin{align}\label{eq:J:p:constrained}
    \set{J}(p_\rv{X}^\star) &= \left\{d_\rv{X}\in\mathbb{R}^\set{X}\, \middle\vert\ \begin{aligned}
    & \sum_{x\in\set{X}} d_\rv{X}(x) = 0 \\
    & d_\rv{X}(x)+p_\rv{X}^\star(x) \geq 0 \quad \forall x\in\set{X}\\
    & \sum_{x} a(x) \cdot \left( d_\rv{X}(x)+p_\rv{X}^\star(x) \right) \geq 0 \quad \forall a\in\set{A} \\
    & \inner{d_\rv{X}}{v} \leq 0 \quad \forall v\in\set{V}\text{ s.t. }
    \begin{aligned}
    &v(x)\geq 0 \ \forall x\in\set{X}_0(p_\rv{X}^\star) \\
    &\inner{v}{a}\geq 0 \ \forall a\in\set{A}_0(p_\rv{X}^\star)
    \end{aligned}
    \end{aligned}\right\},
    \end{align}
    where $\set{A}_0(p_\rv{X}^\star)\defeq\left\{a\in\set{A} \,\middle\vert\,\inner{a}{p_\rv{X}^\star}=0\right\}$.
    The first three constraints for the set ensure that $p_\rv{X}^\star+d_\rv{X}$ is a feasible point in $\set{P}_\set{A}(\set{X})$.
    The last constraint ensures that $p_\rv{X}^\star+d_\rv{X}$ has $p_\rv{X}^\star$ as the closest point in $\Pi^\set{A}$ where $v$ are all valid directions $p^\star_\set{X}$ can move along (for some small distance) without leaving $\Pi^\set{A}$.
\item As a consequence of the above change, we need to rewrite~\eqref{eq:cone:J:p} as
    \begin{align}
    \Cone\left(\set{J}(p_\rv{X}^\star)\right) &= \left\{d_\rv{X}\in\mathbb{R}^\set{X} \,\middle\vert\,\begin{aligned}
    & \sum_{x\in\set{X}} d_\rv{X}(x) = 0 \\
    & d_\rv{X}(x)\geq 0 \quad \forall x\in\set{X}_0(p_\rv{X}^\star)\\
    & \inner{d_\rv{X}}{a}\geq 0 \quad \forall a\in\set{A}_0(p_\rv{X}^\star)
    \end{aligned}\right\} \cap -\left\{v\in\set{V}\middle\vert
    \begin{aligned}
    &v(x)\geq 0 \ \forall x\in\set{X}_0(p_\rv{X}^\star) \\
    &\inner{v}{a}\geq 0 \ \forall a\in\set{A}_0(p_\rv{X}^\star)
    \end{aligned}
    \right\}^*.
    \end{align}
    Again, we can argue that there are finite number of different cones for all $p^\star_\rv{X}$ since above expression only depends on $\set{X}_0(p_\rv{X}^\star)$ and $\set{A}_0(p_\rv{X}^\star)$.
    This ensures the closeness of the cone in~\eqref{eq:cone:J}, which proved Lemma~\ref{lem:cone:closed} under the constrained scenario.
\end{enumerate}
%***********************************************************
\section*{Acknowledgements}
We thank Yury Polyanskiy for helpful discussions about the proof in~\cite{polyanskiy2010channel} and possible ways to complete it. In the end, we found our approach easier to fully formalize. We also thank him for encouraging us to write this comment and suggesting the extensions to the constrained case.
We would also like to thank Hao-Chung Cheng for discussions on closing the gap in Strassen's proof.
MC and MT are supported by NUS startup grants (R-263-000-E32-133 and R-263-000-E32-731).
%***********************************************************
\bibliographystyle{IEEEtran}
\bibliography{reference}
%***********************************************************
\end{document}